\documentclass[a4paper]{article}
\usepackage[utf8]{inputenc}
\usepackage{mathtools}
\usepackage{amsmath}
\usepackage{amsthm}
\usepackage{amssymb} 
\usepackage{graphicx} 
\usepackage{subcaption}
\usepackage{mwe} 
\usepackage[letterpaper,margin=1in]{geometry}
\usepackage{xcolor}
\usepackage[hidelinks]{hyperref}
\usepackage[ruled,vlined,linesnumbered]{algorithm2e}
\usepackage{tikz}

\title{An Improved Random Shift Algorithm for Spanners and Low Diameter Decompositions\footnote{To be presented at the 25th International Conference on Principles of Distributed Systems (OPODIS 2021)}} 
\author{Sebastian Forster, Martin Gr\"osbacher, Tijn de Vos}
\date{University of Salzburg, Austria\\[2ex] \today}

\usepackage{mathtools,amsmath,amsthm,amssymb} 

\let \P\relax
\DeclareMathOperator{\P}{\mathbb{P}}
\DeclareMathOperator{\ID}{\operatorname{ID}}

\newtheorem{theorem}{Theorem}
\newtheorem{lemma}[theorem]{Lemma}
\newtheorem{corollary}[theorem]{Corollary}
\newtheorem{definition}[theorem]{Definition}

\makeatletter
\renewcommand*{\@fnsymbol}[1]{\ensuremath{\ifcase#1\or *\or \dagger\or \ddagger\or
   \mathsection\or \mathparagraph\or \|\or **\or \dagger\dagger
   \or \ddagger\ddagger \else\@ctrerr\fi}}
\makeatother

\begin{document}
\maketitle

\begin{abstract}
Spanners have been shown to be a powerful tool in graph algorithms. Many spanner constructions use a certain type of clustering at their core, where each cluster has small diameter and there are relatively few spanner edges between clusters. In this paper, we provide a clustering algorithm that, given $k\geq 2$, can be used to compute a spanner of stretch $2k-1$ and expected size $O(n^{1+1/k})$ in $k$ rounds in the CONGEST model. This improves upon the state of the art (by Elkin, and Neiman [TALG'19]) by making the bounds on both running time and stretch independent of the random choices of the algorithm, whereas they only hold with high probability in previous results. Spanners are used in certain synchronizers, thus our improvement directly carries over to such synchronizers. Furthermore, for keeping the \emph{total} number of inter-cluster edges small in low diameter decompositions, our clustering algorithm provides the following guarantees. Given $\beta\in (0,1]$, we compute a low diameter decomposition with diameter bound $O\left(\frac{\log n}{\beta}\right)$ such that each edge $e\in E$ is an inter-cluster edge with probability at most $\beta\cdot w(e)$ in $O\left(\frac{\log n}{\beta}\right)$ rounds in the CONGEST model. Again, this improves upon the state of the art (by Miller, Peng, and Xu [SPAA'13]) by making the bounds on both running time and diameter independent of the random choices of the algorithm, whereas they only hold with high probability in previous results.
\end{abstract}

\noindent

\noindent
\textbf{Acknowledgements:} Supported by the Austrian Science Fund (FWF): P 32863-N\\
The first author would like to thank Merav Parter for mentioning that there was still some room for improvement in the spanner construction.

\section{Introduction}
\label{sec:intro}
Clustering has become an essential tool in dealing with large data sets. The goal of clustering data is to identify disjoint, dense regions such that the space between them is sparse. When working with graphs, this translates to partitioning the vertex set into clusters with relatively few edges between clusters such that the clusters satisfy a particular property. One can for example demand that the clusters have low diameter \cite{Awerbuch85,AKPW95,Bartal96,MPX13}, high conductance \cite{GR99,KVV04,ST11,CHKM12,CHZ19,SW19,CS20}, or low effective resistance diameter \cite{AALG18}. In this paper, we focus on the low diameter decomposition and its connection to spanners. Low diameter decompositions are formally defined as follows. 

\begin{definition}
Let $G=(V,E)$ be a weighted graph. A \emph{probabilistic $(\beta,\delta)$-low diameter decomposition} of $G$ is a partition of the vertex set $V$ into subsets $V_1, \dots, V_l$, called \emph{clusters}, such that
\begin{itemize}
	\item each cluster $V_i$ has \emph{strong diameter} at most $\delta$, i.e., $d_{G[V_i]}(u,v) \leq \delta$ for all $u,v\in V_i$\footnote{For $U\subseteq V$, we write $G[U]$ for the graph induced by $U$, i.e., $G[U]:=(U,U\times U \cap E)$.}; 
	\item the probability that an edge $e\in E$ is an \emph{inter-cluster edge} is at most $\beta \cdot w(e)$, i.e., for $e=(u,v)$, the probability that $u\in V_i$ and $v\in V_j$ for $i\neq j$ is at most $\beta \cdot w(u,v)$. 
\end{itemize}
\end{definition}

In an unweighted graph, another typical definition of the low diameter decomposition replaces the second condition with an upper bound on the number of inter-cluster edges \cite{MPX13}. In this fashion, a probabilistic low diameter decomposition has $O(\beta m)$ inter-cluster edges in expectation. 

Originally, low diameter decompositions were developed for distributed models, where they have been proven useful by reducing communication significantly in certain situations \cite{Awerbuch85,AGLP89}. Later, they also have shown to be fruitful in other models; they have been applied in shortest path approximations \cite{Cohen94}, cut sparsifiers \cite{LR99}, and tree embeddings with low stretch \cite{AKPW95,Bartal96,Bartal98}. 

The clustering technique used for computing low diameter decompositions has implicitly been used to develop sparse spanners \cite{BS07, MPVX15, EN19} and synchronizers \cite{Awerbuch85,APPS92}. The main idea is to create the clusters, and add some, but not all, of the inter-cluster edges. In a sense, the inter-cluster edges are sparsified. We formalize this concept as follows. 

\begin{definition}
Let $G=(V,E)$ be an unweighted graph. A \emph{sparsified $(\zeta,\delta)$-low diameter decomposition} of $G$ is a partition of the vertex set $V$ into subsets $V_1, \dots, V_l$, called \emph{clusters}, together with a set of edges $F\subseteq E$ such that
\begin{itemize}
	\item each cluster $V_i$ has \emph{strong diameter} at most $\delta$, i.e., $d_{G[V_i]}(u,v) \leq \delta$ for all $u,v\in V_i$; 
	\item for every edge, one of its endpoints has an edge from $F$ into the cluster of the other endpoint, i.e., $\forall e=(u,v) \in E$, we have either $(u',v)\in F$ for some $u' \in C_u$ or $(u,v')\in F$ for some $v' \in C_v$\footnote{For $v\in V$, we write $C_v$ for the cluster containing $v$.}. 
	\item  $|F|\leq \zeta$.
\end{itemize}
Moreover, we say that a sparsified $(\zeta,\delta)$-low diameter decomposition is \emph{tree-supported} if for each cluster $V_i$ we have a cluster center $c_i\in V_i$ and a tree of height at most $\delta/2$ spanning the cluster. All these trees together are called the \emph{support-forest}. 
\end{definition}

Our main result is a clustering algorithm that produces a sparsified low diameter decomposition. 

\begin{theorem}
\label{thm:sparsifiedLDD}
There exists an algorithm, such that for each unweighted graph $G=(V,E)$ and parameter $k\geq 2$ it outputs a tree-supported sparsified $\left(\zeta,2k-2\right)$-low diameter decomposition, with $\zeta=O(n^{1+1/k})$ in expectation. The algorithm runs in $k$ rounds in the CONGEST model, and in $O(k \log^*n)$ depth and $O(m)$ work in the PRAM model.
\end{theorem}

An important feature of this result is that the bounds on the strong diameter and number of rounds are not probabilistic; they are independent of the random choices in the algorithm. We show two applications of this theorem: constructing spanners and constructing synchronizers. 

\subparagraph*{Spanners}
Given a graph $G=(V,E)$, we say that $H\subseteq G$ is a \emph{spanner} of stretch $\alpha$, if $d_H(u,v)\leq \alpha\cdot d_G(u,v)$, for every $u,v\in V$. It is straightforward that a tree-supported sparsified $\left(\zeta,2k-2\right)$-low diameter decomposition gives a spanner of size $\zeta+n$ and stretch $2k-1$, for details we refer to Section~\ref{subsec:spanner}. This gives us the following corollary. 

\begin{corollary}\label{cor:spanner}
There exists an algorithm, such that for each unweighted graph $G=(V,E)$ and parameter $k\geq 2$ it outputs a spanner $H$ of stretch $2k-1$. The expected size of $H$ is at most $O(n^{1+1/k})$.  
The algorithm runs in $k$ rounds in the CONGEST model, and in $O\left(k \log^*n\right)$ depth and $O(m)$ work in the PRAM model.
\end{corollary}

Spanners themselves have been useful in computing approximate shortest paths \cite{ABCP93,Cohen98}, distance oracles and labeling schemes \cite{TZ05,Peleg00}, and routing \cite{PU89}. A simple greedy algorithm~\cite{AGDJS93} gives a spanner of stretch $2k-1$ and of size $O(n^{1+1/k})$, which is an optimal trade-off under the girth conjecture \cite{FS16}. However, its fastest known implementation in the RAM model takes $O(k n^{2+1/(2k+1)})$ time \cite{RZ04}. Halperin and Zwick \cite{HZ93} gave a linear-time algorithm to construct spanners with an optimal trade-off for unweighted graphs in the RAM model. However, this algorithm does not adapt well to distributed and parallel models of computation. 
This problem can be overcome by exploiting the aforementioned relation with sparsified low diameter decompositions. This was (implicitly) done by Baswana and Sen \cite{BS07}, who provide an algorithm that computes a spanner of stretch $2k-1$ and of size $O(kn^{1+1/k})$ in $O(k)$ rounds. The state of the art is by Elkin and Neiman \cite{EN19}, which builds off \cite{MPVX15}, and is also based on low diameter decompositions. They provide an algorithm that with probability $1-1/c$ computes a $(2k-1)$-spanner of expected size $O(c^{1/k}n^{1+1/k})$ in $k$ rounds. Standard techniques of boosting the failure probability to something inverse polynomial (or `with high probability') will require a logarithmic overhead. Alternatively, one can view the algorithm of Elkin and Neiman as an algorithm that outputs an $\alpha$-spanner of expected size $O(c^{1/k}n^{1+1/k})$ in $O(\alpha)$ rounds, such that with probability $1-1/c$ we have that $\alpha=2k-1$. 

Corollary~\ref{cor:spanner} improves on the result of Elkin and Neiman by making the bounds on the stretch and the running time independent of the random choices in the algorithm. In particular, the algorithm of Elkin-Neiman involves sampling vertex values from an exponential distribution. The exponential distribution introduces an (as we show) unnecessary amount of randomness; we demonstrate that the geometric distribution suffices. We replace the extra random bits the exponential distribution provides by a tie-breaking rule on the vertex $\ID$s, which we believe contributes to a more intuitive construction. 

\subparagraph*{Synchronizers}
The second application of Theorem~\ref{thm:sparsifiedLDD} is in constructing synchronizers in the CONGEST model. A synchronizer gives a procedure to run a synchronous algorithm on an asynchronous network. More precisely, the goal is to run any synchronous $R(n)$-round $M(n)$-message complexity CONGEST model algorithm on an asynchronous network with minimal time and message overhead. The first results on synchronizers are by Awerbuch \cite{Awerbuch85}, called synchronizers $\alpha$, $\beta$, and $\gamma$. Subsequently, these results were improved by Awerbuch and Peleg \cite{AP90}, and Awerbuch et al.\ \cite{APPS92}, both having $O(R(n)\log^3 n)$ time and $O(M(n)\log^3 n)$ message complexity. 

The synchronizer $\gamma$ from \cite{Awerbuch85} essentially consists of running a combination of the simple synchronizers $\alpha$ and $\beta$ on a sparsified $(\zeta,\delta)$-low diameter decompositions. In that case, the bound $\zeta$ on the sparsified inter-cluster edges goes into the bound for the communication overhead of the synchronizer and the strong diameter bound $\delta$ goes into the bound for the time overhead of the synchronizer. Applying synchronizer $\gamma$ on our clustering, we obtain the following result.

\begin{theorem}\label{thm:synchronizer}
There exists an algorithm that, given $k\geq 2$, can run any synchronous $R(n)$-round $M(n)$-message complexity CONGEST model algorithm on an asynchronous CONGEST network. In expectation, the algorithm uses a total of $O(M(n)+R(n)n^{1+1/k})$ messages. Provided that each message incurs a delay of at most one time unit, it takes $O(R(n)k)$ rounds. The initialization phase takes $O(k)$ time, using $O(km)$ messages.  
\end{theorem}

The running time claimed in this theorem is independent of the random choices in our algorithm, which is a direct result of Theorem~\ref{thm:sparsifiedLDD}. The previous sparsified low diameter decompositions (implicit in \cite{EN19}) would provide similar bounds on the running time, but only with constant probability. 

\subparagraph*{Low Diameter Decompositions}
Perhaps unsurprisingly, we show that, with the right choice of parameters, our clustering algorithm can also compute \emph{unsparsified} low diameter decompositions. 

\begin{theorem} \label{thm:LDD}
There exists an algorithm, such that for each graph $G=(V,E)$ with integer weights $w\colon E \to \{1,\dots, W\}$ and parameter $\beta\in(0,1]$ it outputs a low diameter decomposition, whose components are clusters of strong diameter of at most $O\left(\frac{\log n}{\beta}\right)$. Moreover, each edge is an inter-cluster edge with probability at most $\beta\cdot w(u,v)$. 
The algorithm runs in $O\left(\frac{\log n}{\beta}\right)$ rounds in the CONGEST model, and in $O\left(\frac{\log n \log^*n}{\beta}\right)$ depth and $O(m)$ work in the PRAM model.
\end{theorem}

Similar to our spanner algorithm, the bounds on the running time and strong diameter hold independent of the random choices within the algorithm, as opposed to the previous state of the art \cite{MPX13}, where they only hold with high probability. 
In the low diameter decomposition as discussed above, the trade-off between $\beta$ and diameter bound $O\left(\frac{\log n}{\beta}\right)$ is essentially optimal \cite{Bartal96}.

\subsection*{Technical Overview}
Our clustering algorithm follows an approach known as ball growing, related to the probabilistic partitions of \cite{Bartal96,Bartal98}. In a sequential setting, this consists of picking a vertex, and repeatedly adding the neighbors of the current vertices to the ball. This stops when a certain bound is reached, such as a bound on the diameter of the ball or on the number of edges between the current ball and the remainder of the graph. The algorithm repeats this procedure with the remainder of the graph until this is empty. Miller, Peng, and Xu \cite{MPX13} showed that this can be parallelized by letting each vertex create its own ball, but after a certain start time delay. In \cite{MPX13}, this has been done by sampling the delays from the exponential distribution, which leads to the aforementioned probabilistic diameter guarantee, as the exponential distribution can take infinitely high values -- albeit with small probability. Furthermore, multiple authors (see e.g.\ \cite{FG19,MPX13}) argue that one can round the sampled values from the exponential distribution for most of the algorithm and solely use that the fractional values of the sampled value induce a random permutation of the nodes. In this paper, we show that even fewer random bits are needed: we do not require a random permutation of the nodes. We demonstrate that a tie-breaking rule based on the $\ID$s is enough. 

We sample with a capped-off geometric distribution, also used in \cite{LS93,APPS92}. As opposed to the standard geometric distribution, the capped-off version can only take a finite number of values. We believe this leads to a more direct proof of the spanner algorithm of \cite{EN19} and of the decomposition algorithm of \cite{MPX13}. Moreover, by making the sparsifier low diameter decomposition explicit, the application to synchronizers is almost immediate. In the remainder of this paper, we will not think of the sampled values as start time delays, but as the distance to some conceptual source $s$, similar to the view in~\cite{MPX13}. The rest of the clustering algorithm then consists of computing a shortest path tree rooted at $s$, which is easily calculated, both in the CONGEST and PRAM model. The clusters consists of the trees that remain when we disconnect the shortest path tree by removing the root $s$. 

As an anonymous reviewer pointed out, in the case of low diameter decompositions, the algorithm of Miller et al.~\cite{MPX13} admits an alternative approach. We can exploit the fact that the exponential delays are bounded with high probability. In case the delays exceed the bound, we could return a sub-optimal clustering, without any central communication. As this only happens with low probability, it does not impact the expected number of inter-cluster edges. Note however, that the spanner construction of Elkin and Neiman~\cite{EN19} is not in this high-probability regime, therefore this straightforward approach would not work. We additionally believe that, beyond the result itself, our algorithm provides a more streamlined view.

\section{The Clustering Algorithm}\label{sec:clustering}
Let $G=(V,E)$ be a graph with integer weights $w\colon E \to \{1,\dots, W\}$. Let $p\in (0,1)$ and $r\in \mathbb N$ be parameters, to be chosen according to the application of our algorithm. In the following, we provide an algorithm for computing a clustering, where the strong diameter of these clusters will be $2r$. In particular, we will show that each cluster is tree-supported by a tree of height $r$. The number of inter-cluster edges depends on both $p$ and $r$, and can be bounded in two ways. The first approach, detailed in Section~\ref{sec:spanner}, shows we have a sparsified low diameter decomposition. Here, for each vertex we compute the expected number of edges in the sparsified set of inter-cluster edges, which gives a bound that does not depend on $m$, but only on $n,p$ and~$r$. The second approach, detailed in Section~\ref{sec:LDD}, shows we have a probabilistic low diameter decomposition, by computing the probability that any edge is an inter-cluster edge.

\subsection{Construction}
First we conceptually add a node $s$ to the graph $G$ to form the graph $G'$. The node $s$ will function as an artificial source for a shortest path tree. Each vertex will have a distance to $s$ in $G'$ depending on some random offset. Hereto, each vertex samples a value $\delta_v$ from the \emph{capped exponential distribution} $\text{GeomCap}(p,r)$, defined by
\begin{align*}
\P\left[ \text{GeomCap}(p,r) =i \right] = 
\begin{cases}
p(1-p)^i &\text{if } 0\leq i\leq r-1; \\ 
(1-p)^{r} &\text{if } i=r; \\
0 &\text{else.}
\end{cases}
\end{align*}
This distribution corresponds to the model where we repeat at most $r$ Bernoulli trials, and measure how many trials occur (strictly) before the first success, or whether there is no success in the first $r$ trials. We check that $\text{GeomCap}$ is indeed a probability distribution on $\{0,1,\dots, r\}$:
\begin{align*}
\sum_{i=0}^{r} \P\left[ \text{GeomCap}(p,r) =i \right] = \sum_{i=0}^{r-1}p(1-p)^i  + (1-p)^r =p \frac{1-(1-p)^{r}}{1-(1-p)}+(1-p)^{r} =1.
\end{align*}
As the intuition suggests, $\text{GeomCap}$ has a memoryless property as long as the cap is not reached, i.e., $\P\left[ \text{GeomCap}(p,r) =i \;\middle|\; \text{GeomCap}(p,r) \geq i\right]=p$ for $i\leq r-1$. The proof is completely analogous to the proof of the memoryless property of the geometric distribution.

For each vertex $v\in V$, we conceptually add an edge $(s,v)$ to $G'$, with weight $w(s,v):= r-\delta_v$. We define $d^{(u)}(s,x) := w(s,u) + d_G(u,x)$, which is the minimal path length, for a path from $s$ to $x$ over $u$. Now we have that the distance from $s$ to $x$ equals $d_{G'}(s,x) = \min_{u\in V} \{ d^{(u)}(s,x)\}$. We call this the \emph{level} of $x$, ranging from $0$ (closest to $s$) to $r$ (furthest from $s$). Moreover, we define $p_u(x)$ to be the predecessor of $x$ on an arbitrary but fixed shortest path from $u$ to $x$. 
Next, we construct a shortest path tree $T$ rooted at $s$. When necessary, we do tie-breaking according to $\operatorname{ID}$s: let $v$ be such that $d_{G'}(s,x) = d^{(v)}(s,x)$ and $\ID(v) < \ID(u)$ for all $u\in V$ satisfying $d_{G'}(s,x) = d^{(u)}(s,x)$. Then we connect $x$ to the shortest path tree using the edge $(p_v(x),x)$. Moreover, we add $x$ to the cluster of $v$ and write $c_x :=v$ for the corresponding cluster center. Intuitively, the clusters correspond to the connected components that remain when we remove the source $s$ from the created shortest path tree. The formal argument for this can be found in the proof of Lemma~\ref{lm:tree-supported}. The computation of this shortest path tree is model-specific, we provide details in Section \ref{sec:implementation}. 

The algorithm outputs the shortest path tree $T$, and for each $x\in V$, the center of its cluster center $c_x$ and its level. The knowledge of cluster centers immediately gives a clustering, where -- by the remark above -- each cluster has radius at most $r$. In Section~\ref{sec:spanner}, we show how to construct a set of edges $F\subseteq E$ from the cluster centers and levels, such that $H:=T\cup F$ is a spanner. 

In the above, we only need an arbitrary ordering of the vertices. If we assume that each vertex $v\in V$ has a unique identifier, $\operatorname{ID}(v) \in \{1,\dots, N\}$, we can provide an alternative way of constructing the same shortest path tree. We construct a graph where $w(s,v) = r-\delta_v+\operatorname{ID}(v)/(N+1)$, and compute a shortest path tree rooted at $s$. This embeds the tie-breaking rule in the weight of the added edges, and thus in the distances. For generality -- and suitable implementation in distributed models with limited bandwidth -- the remainder of this paper relies on the former characterization using the tie-breaking rule.

\subsection{Tree-Support}
Next, we will show that the created clusters are tree-supported by a tree of height $r$. We have already chosen cluster centers, and we will show that we can identify trees rooted at these centers that satisfy the tree-support condition. 

\begin{lemma}\label{lm:tree-supported}
Each cluster is tree-supported by a tree of height at most $r$. 
\end{lemma}
\begin{proof}
Let $v\in V$ be a vertex, which is part of the cluster centered at $c_v$. We show that there is a path from $c_v$ to $v$ contained in this cluster, which has length at most $d_G(c_v,v)\leq r$.We proceed to show by induction on $d_G(c_v,v)$ that there is a path from $c_v$ to $v$ contained in their cluster, which has length at most $d_G(c_v,v)$. The base case, $v=c_v$, is trivial. Let $u$ be the predecessor of $v$ on some path from $c_v$ to $v$ of length $d_G(c_v,v)$. It suffices to show that $u$ is in the same cluster, then the result follows from the induction hypothesis. By definition of $c_v$, we have that $$d_{G'}(s,v) = d_{G'}(s,c_v)+d_G(c_v,v)=d_{G'}(s,c_v)+d_G(c_v,u) + w(u,v) = d^{(c_v)}(s,u) + w(u,v).$$ By the triangle inequality we have $d_{G'}(s,v) \leq d_{G'}(s,u)+w(u,v)$. Combining this, we see $d^{(c_v)}(s,u) \leq  d_{G'}(s,u)$. As the distance $d_{G'}(s,u)$ is minimal, by definition we have $d^{(c_v)}(s,u) =  d_{G'}(s,u)$. Now suppose that $u$ is part of some cluster $c_u$. Then we have $d^{(c_u)}(s,u) =  d_{G'}(s,u)$ and $\ID(c_u) \leq \ID(c_v)$. However, this implies that  $d^{(c_u)}(s,v) \leq  d^{(c_u)}(s,u)+w(u,v) =  d^{(c_v)}(s,u) +w(u,v)=d^{(c_v)}(s,v)$. Hence by the tie-breaking for $v$ we have $\ID(c_v) \leq \ID(c_u)$ and thus $c_u = c_v$. 
\end{proof}
As an immediate corollary, we obtain a bound on the strong diameter of the clusters. 
\begin{corollary}\label{cor:strongdiameter}
Each cluster has a strong diameter of $2r$. 
\end{corollary}

\subsection{Implementation and Running Time}
\label{sec:implementation}
For the RAM model, the implementation is straightforward and can be done in linear time \cite{Thorup99}. The implementation in distributed and parallel models requires a little more attention. For both models, the computational aspect is very similar to prior work \cite{MPX13,EN19}.

\subsubsection{Distributed Model}
The algorithm as presented, can be implemented efficiently both in the LOCAL and in the CONGEST model. It runs in $r+1$ rounds as follows. In the initialization phase, each vertex~$v$ samples its value~$\delta_v$ and sets its initial distance to the conceptual vertex~$s$ as $r-\delta_v$. In the first round of communication, $v$ sends the tuple $(r-\delta_v,\operatorname{ID}(v) )$ to its neighbors. In each round, $v$ updates its distance to $s$ according to received messages. It then broadcasts the tuple of its updated distance and the $\operatorname{ID}$ corresponding to the first vertex on the path from $s$ to $v$. Note that at the end of the algorithm, each node knows its own level and cluster center, and the level and cluster center of each of its neighbors. 

When the algorithm is applied with $r = O(n)$ (if $r\geq n$, we can simply return the connected components of the graph as clusters), we maintain a bound on the message size of $O(\log n)$, so there are no digit precision consideration for the CONGEST model.
Moreover, each vertex $v$ has distance at most $r-\delta_v\leq r$ to $s$, the algorithm terminates within $r+1$ rounds.\footnote{The `$+1$' appears, as nodes in the lowest level have distance $0$ to the source $s$.}

\subsubsection{PRAM Model}
The implementation in the PRAM model is slightly different to the CONGEST model. Instead of broadcasts by each vertex in each round, a vertex $v$ updates its distance only once: either after one of its neighbors updated its distance, or after time $r+1-\delta_v$ it sets its distance to $r-\delta_v$. The total required depth differs on the exact model of parallelism, it is $O(r\log^* n)$ in the CRCW model of parallel computation. To show this, we follow the general lines of \cite{KS97}, but we have to be careful: during the shortest path computation, we might need to apply our tie-breaking rule, i.e., finding the minimum $\ID$ among all options. Note that in the PRAM model, we can assume without loss of generality that the $\ID$s are labeled $1$ to $n$ in the adjacency list representation. Finding the minimum can be done with high probability in $O(\log^*n)$ depth and $O(n)$ work, as we can sort a list of integers between $1$ and $n$ in $O(\log^*n)$ depth and $O(n)$ work \cite{GMV91}. If we exceed the $O(\log^*n)$ depth bound, we stop and output the trivial clustering consisting of singletons. This clustering clearly satisfies the diameter bound, and as we only output it with low probability, it has no effect on the expected number of inter-cluster edges. So we can conclude that the additional sorting overhead for the tie-breaking is a factor $O(\log^*n)$. The algorithm has total work $O\left(m+\frac{n}{1-p}\right)$, where the contribution of $O\left(\frac{n}{1-p}\right)$ comes from sampling from the geometric distribution. In this paper, this factor vanishes as we always have $p$ such that $O\left(\frac{n}{1-p}\right)=O(m)$.

\section{Constructing a Sparsified Low Diameter Decomposition}
\label{sec:spanner}
In this section, we show how the clustering algorithm leads to a sparsified low diameter decomposition. The procedure is as follows: given $k\geq 2$, we set $r=k-1$, $p = 1-n^{-1/k}$, and compute a clustering using the algorithm of Section~\ref{sec:clustering}. We denote $F\subseteq E$ for the sparsified set of inter-cluster edges. Intuitively, for each vertex $v\in V$, we add an edge to $F$ for each cluster in which we have a neighbor on one level below, or a neighbor on the same level as $v$ with the $\ID$ of the cluster center smaller than the $\ID$ of center of the cluster of $v$. 

\begin{lemma}
\label{lm:spannersize}
There exists a set of edges $F\subseteq E$ of expected size $O(n^{1+1/k})$, such that for every edge, one of its endpoints has an edge from $F$ into the cluster of the other endpoint.  
\end{lemma}
\begin{proof} 
We define, $F:=\bigcup_{x\in V}C(x)$, where $C(x)$ consists of the following edges
\begin{align*}
C(x) := &\{ (x,p_u(x) ) : d^{(u)}(s,x) = d_{G'}(s,x)   \}\\
&\cup \{ (x,p_u(x) ) : d^{(u)}(s,x) = d_{G'}(s,x)+1 \text{ and } \ID(u)<\ID(c_x)   \}. 
\end{align*}
First, we show that $F$ satisfies the property stated in the lemma, then we consider its size. Let $(x,y)\in E$, without loss of generality, we assume $d_{G'}(s,x)\geq d_{G'}(s,y)$, and in case of equality we assume $\ID(c_x)>\ID(c_y)$. We will show that there is an edge $(x,p_{c_y}(x))\in C(x)$ to the cluster of $y$. First of all, notice that $d_{G'}(s,y) \geq d_{G'}(s,x)-1$ by the triangle inequality. If $d_{G'}(s,y) = d_{G'}(s,x)-1$, then $d^{(c_y)}(s,x) \leq d_{G'}(s,x)$. Because of minimality of $d_{G'}(s,x)$, we have $d^{(c_y)}(s,x) = d_{G'}(s,x)$, and thus $(x,p_{c_y}(x))\in C(x)$ by definition of $C(x)$. If $d_{G'}(s,y) = d_{G'}(s,x)$, we have $\ID(c_x)>\ID(c_y)$. Moreover, we have $d^{(c_y)}(s,x) \leq d_{G'}(s,x)+1$. So again it follows that $(x,p_{c_y}(x))\in C(x)$ by definition of $C(x)$.

Now, we turn to the expected size of $F$. By linearity of expectation, we have $\mathbb E\left[F\right]= \sum_{x\in V} \mathbb E\left[C(x)\right]$. We will show that for each $x\in V$ the expected size of $C(x)$ is at most $2n^{1/k}$. For each $u\in V$, we potentially add an edge $(x,p_u(x))$ to $C(x)$.  First, we calculate the probability that at least $t$ such vertices $u$ contribute an edge. Hereto, we look at the random variables $X_u=d^{(u)}(s,x)=k-\delta_u+d_G(u,x)$. According to these random variables, we order all vertices: $V=\{u_1,u_2,\dots, u_n\}$, such that for $i<j$ we satisfy one of the following properties
\begin{itemize}
\item $X_{u_i} < X_{u_j}$;
\item $X_{u_i} = X_{u_j}$ and $\ID(u_i)<\ID(u_j)$.
\end{itemize}
We calculate $\P[|C(X)|\geq t]$, i.e., the probability that $\{(x,p_{u_1}(x)), \dots, (x,p_{u_t}(x))\}\subseteq C(x)$. We do this by conditioning on $u_t=v$. We observe
\begin{align*}
\P\left[ \{(x,p_{u_1}(x)), \dots, (x,p_{u_t}(x))\}\subseteq C(x) \;\middle|\; u_t=v \right] =\\ \prod_{i=1}^{t-1} \P\left[ (x,p_{u_i}(x))\in C(x) \text{ and } (x,p_v(x))\in C(x)\;\middle|\; u_t=v\right].
\end{align*}
So we calculate $\P\left[ (x,p_{u_i}(x))\in C(x) \text{ and } (x,p_v(x))\in C(x)\;\middle|\; u_t=v\right]$ for $i<t$. By definition of $C(x)$, this can only hold if either $x$'s closest neighbors in the clusters centered at $u_i$ and $v$ are on the same level (in which case we have $\ID(u_i)<\ID(v)$, as $v=u_t$ and $i<t$) or the neighbor from the cluster centered at $u_i$ is at a level lower and $\ID(v)<\ID(u_i)$. Note that the level of the closest neighbor in the cluster of $u_i$ or $v$ corresponds to the distance $d^{(u_i)}(s,x)=X_{u_i}$ or $d^{(v)}(s,x)$ respectively. As the allowed distances depend on the $\ID$ of $u_i$, we split the vertices according to $\ID$:
\begin{align*}
V_< &:= \{ u\in V : \ID(u)<\ID(v) \}; \\
V_> &: = \{u\in V : \ID(u)>\ID(v) \}.
\end{align*}
If $u_i \in V_<$, then we know $X_{u_i}\leq d^{(v)}(s,x)$. If both $(x,p_{u_i}(x))$ and $(x,p_v(x))$ are in $C(x)$, we must have $X_{u_i}=d^{(v)}(s,x)$. So for every $i<t$, we are looking at
\begin{align*} 
&\P\left[ (x,p_{u_i}(x))\in C(x) \text{ and } (x,p_v(x))\in C(x)\; \middle|\; u_i \in V_<\text{ and }u_t=v\right] \\
&\leq \P\left[ X_{u_i} = d^{(v)}(s,x) \; \middle|\; u_i \in V_<\text{ and } u_t=v \right],\\
&= \P\left[ X_{u_i} = d^{(v)}(s,x) \; \middle|\; u_i \in V_<\text{ and } X_{u_i} \leq d^{(v)}(s,x) \right],
\end{align*}
where the last equality holds as we only rewrote the condition using the order of the $u_j$'s. We fill in the definition of $X_{u_i}$ and use that the probability of the event we are looking at is independent of $\ID(u_i)$:
\begin{align*} 
&\P\left[ X_{u_i} = d^{(v)}(s,x) \; \middle|\; u_i \in V_<\text{ and } X_{u_i} \leq d^{(v)}(s,x) \right] \\
&= \P\left[k-\delta_{u_i}+d_G(u_i,x) = d^{(v)}(s,x) \; \middle|\; u_i \in V_<\text{ and } k-\delta_{u_i}+d_G(u_i,x) \leq d^{(v)}(s,x) \right]\\
&= \P\left[k-\delta_{u_i}+d_G(u_i,x) = d^{(v)}(s,x) \; \middle|\; k-\delta_{u_i}+d_G(u_i,x) \leq d^{(v)}(s,x) \right]\\
&= \P\left[\delta_{u_i} = k+d_G(u_i,x)-d^{(v)}(s,x) \; \middle|\;\delta_{u_i} \geq k+d_G(u_i,x)-d^{(v)}(s,x) \right].
\end{align*} 
When $ k+d_G(u_i,x)-d^{(v)}(s,x) \leq k-2$, this equals $p$, by the memoryless property of the geometric distribution. To distinguish this, we partition the vertices $u$ with $\ID(u)<\ID(v)$ into two set
\begin{align*}
V_{<,1} &:= \{ u \in V: \ID(u)<\ID(v) \text{ and } k+d_G(u_i,x)-d^{(v)}(s,x) \leq k-2\}; \\
V_{<,2} &: = \{u \in V: \ID(u)<\ID(v) \text{ and } k+d_G(u_i,x)-d^{(v)}(s,x) = k-1\}.
\end{align*}
Now for $u_i \in V_>$, we obtain by the same reasoning
\begin{align*} 
&\P\left[ (x,p_{u_i}(x))\in C(x) \text{ and } (x,p_v(x))\in C(x)\; \middle|\; u_i \in V_>\text{ and }u_t=v\right] \\
&= \P\left[\delta_{u_i} = k+d_G(u_i,x)-d^{(v)}(s,x)+1 \; \middle|\;\delta_{u_i} \geq k+d_G(u_i,x)-d^{(v)}(s,x)+1 \right].
\end{align*} 
As before, when $ k+d_G(u_i,x)-d^{(v)}(s,x) +1 \leq k-2$, this equals $p$, by the memoryless property of the geometric distribution. And again, we partition $V_>$ into two sets
\begin{align*}
V_{>,1} &:= \{ u \in V: \ID(u)>\ID(v) \text{ and } k+d_G(u_i,x)-d^{(v)}(s,x)+1 \leq k-2\}; \\
V_{>,2} &: = \{u \in V: \ID(u)>\ID(v) \text{ and } k+d_G(u_i,x)-d^{(v)}(s,x) +1= k-1\}.
\end{align*}
If we define $V_1 = V_{<,1} \cup V_{>,1}$ and $V_2 = V_{<,2} \cup V_{>,2}$, we can summarize our results as
\begin{align*}
\P\left[ (x,p_{u_i}(x))\in C(x) \text{ and } (x,p_v(x))\in C(x)\; \middle|\; u_i \in V_1 \text{ and }u_t=v\right] &\leq p.
\end{align*}
Next, we split the expected value of $C(x)$ depending on $V_1$ and $V_2$:
\begin{align*}
 \mathbb{E}\left[ |C(x)| \;\middle|\; u_t =v \right] &=  \mathbb{E}\left[ |C(x)\cap V_{1}| \;\middle|\; u_t =v \right]+ \mathbb{E}\left[ |C(x)\cap V_{2}| \;\middle|\; u_t =v \right].
\end{align*}
We bound the first summand with $n^{1/k}$, independent of $v$. Hereto, we observe that for any non-negative discrete random variable $X$ we have
\begin{align*}
\mathbb{E}[X] = \sum_{s=1}^\infty s \P[X=s] =\sum_{s=1}^\infty \sum_{t=1}^s \P[X=s] = \sum_{t=1}^\infty \sum_{s=t}^\infty \P[X=s]  = \sum_{t=1}^\infty \P[X\geq t].
\end{align*}
Using this, we obtain
\begin{align*}
\mathbb{E}\left[ |C(x)\cap V_{1}|\;\middle|\; u_t = v \right] &= \sum_{t=1}^n \P\left[ |C(x)\cap V_{1}|\geq t \;\middle|\; u_t =v \right] \\
&\leq \sum_{t=1}^n \prod_{i=1}^{t-1}\P\left[ (x,p_{u_i}(x))\in C(x) \text{ and }(x,p_v(x))\in C(x) \;\middle|\; u_i \in V_1 \text{ and } u_t =v \right] \\
&\leq \sum_{t=1}^n p^{t-1} \\
&\leq \sum_{t=0}^\infty p^{t} \\
&= \frac{1}{1-p}\\
&= n^{1/k},
\end{align*}
where the last equality holds by definition of $p$. For the second summand, we look at all $v$ simultaneously.
\begin{align*}
	\sum_{v\in V}\mathbb{E}\left[ |C(x)\cap V_2| \;\middle|\; u_t = v \right]\P\left[u_t =v\right] &\leq \sum_{v\in V}\mathbb{E}\left[ |V_2| \;\middle|\; u_t = v \right]\P\left[u_t =v\right] \\
&\leq \sum_{v\in V}\mathbb{E}\left[ |\{ u \in V: \delta_u = k-1\}| \;\middle|\; u_t = v \right]\P\left[u_t =v\right] \\
&= \mathbb{E}\left[ |\{ u \in V: \delta_u = k-1\}| \right],
\end{align*}
where the last equality holds by the law of total probability. We bound this as follows
\begin{align*}
 \mathbb{E}\left[ |\{ u \in V: \delta_u = k-1\}| \right] &\leq n \P\left[ \delta_u = k-1\right] = n(1-p)^{k-1} = n^{1/k},
\end{align*}
where the last equality holds by definition of $p$. In total, this gives us $\mathbb{E}[ |C(x)|]\leq 2 n^{1/k}$. 
\end{proof}

Together Lemma~\ref{lm:tree-supported} and Lemma~\ref{lm:spannersize} imply the following theorem. 

\begingroup
\def\thetheorem{\ref{thm:sparsifiedLDD}}
\begin{theorem}[Restated]
There exists an algorithm, such that for each unweighted graph $G=(V,E)$ and parameter $k\geq 2$ it outputs a tree-supported sparsified $\left(\zeta,2k-2\right)$-low diameter decomposition, with $\zeta=O(n^{1+1/k})$ in expectation. The algorithm runs in $k$ rounds in the CONGEST model, and in $O(k \log^*n)$ depth and $O(m)$ work in the PRAM model.
\end{theorem}
\addtocounter{theorem}{-1}
\endgroup

\subsection{Constructing a Spanner}\label{subsec:spanner}
Now, we can construct a spanner from the tree supported low diameter decomposition in the following manner. Let $T$ denote the support forest, and let $F$ denote the set as given in Lemma~\ref{lm:spannersize}. We define the spanner $H:=(V,T\cup F)$. As any forest has at most $n-1$ edges, the expected size of $H$ is at most $O(n^{1+1/k})$. Actually, one could also show that $T\subseteq F$ in our construction of $F$, but this would not impact the asymptotic size bound. To show that $H$ is a spanner, we show its of limited stretch. 

\begin{lemma}\label{lm:stretch}
$H$ is a spanner of stretch $2k-1$.
\end{lemma}
\begin{proof}
We will show that for every edge $(x,y)\in E$, there exists a path from $x$ to $y$ in $H$ of length at most $2k-1$. Consequently we have that $d_H(x,y) \leq (2k-1)d_G(x,y)$ for every $x,y\in V$, hence $H$ is a spanner of stretch $2k-1$. 

Let $(x,y) \in E$. By definition of $F$, one of the endpoints has an edge in $F$ into the cluster of the other endpoint. Without loss of generality, let there be an edge $(x,z)\in F$ with $z$ in the cluster of $y$. By Corollary~\ref{cor:strongdiameter}, there is path of length at most $2(k-1)$ from $z$ to $y$, so in total we have a path of length at most $2(k-1)+1=2k-1$ from $x$ to $z$ to $y$. 
\end{proof}

Now, the following corollary follows from Theorem~\ref{thm:sparsifiedLDD} and Lemma~\ref{lm:stretch}.
\begingroup
\def\thetheorem{\ref{cor:spanner}}
\begin{corollary}[Restated]
There exists an algorithm, such that for each unweighted graph $G=(V,E)$ and parameter $k\geq 2$ it outputs a spanner $H$ of stretch $2k-1$. The expected size of $H$ is at most $O(n^{1+1/k})$.  
The algorithm runs in $k$ rounds in the CONGEST model, and in $O\left(k \log^*n\right)$ depth and $O(m)$ work in the PRAM model.
\end{corollary}
\addtocounter{theorem}{-1}
\endgroup

\subsection{Constructing a Synchronizer}\label{subsec:synchronizer}
Suppose we are given a synchronous CONGEST model algorithm, but we want to run it on an asynchronous CONGEST network. That is, the messages sent in the network can now have arbitrary delays and, in an event-driven manner, nodes become active each time they receive a message. For the purpose of analyzing the time complexity of the algorithm, it is often assumed that the delay is at most one unit of time, however, the algorithm should behave correctly under \emph{any} finite delays. In this situation, a node should start simulating its next (synchronous) round when it has received all the messages from the previous round from its neighbors. The problem is that it cannot tell the difference between the situation if a message from a particular neighbor has not arrived yet or if this same neighbor is not sending any message in that round at all. We say that a node is \emph{safe} if all the messages it has sent have arrived at their destination. In order to determine whether all neighboring nodes are safe, additional messages are sent. The procedure governing these additional messages is called the \emph{synchronizer}. There are two things to take into account when analyzing synchronizers. First, the time overhead: how much time is needed to send the additional messages for each synchronous round. Second, the message-complexity (or communication) overhead: how many additional messages are sent. For more details on synchronizers see e.g.\ \cite{Lynch96,KS11}. 

Let us first consider two simple synchronizers: synchronizer $\alpha$ and synchronizer $\beta$, see~\cite{Awerbuch85}. In synchronizer $\alpha$, when a node receives a message from a neighbor, it sends back an `acknowledge' message. When a node has received acknowledge messages for all its sent messages, it marks itself safe and reports this to all its neighbors. The synchronizer $\alpha$ uses, for each simulated synchronous round, additional $O(1)$ time, and $O(m)$ messages. 

Synchronizer $\beta$ will produce a different trade off between time and message overhead. It uses an initialization phase in which it creates a rooted spanning tree, where the root is declared the leader. Now after sending messages of a certain synchronous round, again nodes that receive messages reply with an acknowledge message to each. Nodes that are safe, and whose children in the constructed tree are also safe communicate this to their parent in the tree. Eventually the leader will learn that the whole graph is safe, and will broadcast this along the spanning tree. Synchronizer $\beta$ uses $O(D)$ time and $O(n)$ messages per synchronous round. 

Now we are ready to consider a little more involved example, called synchronizer $\gamma$, see~\cite{Awerbuch85}. This synchronizer makes use of clustering, where within each cluster synchronizer $\beta$ is used and between clusters synchronizer $\alpha$ is used. In the LOCAL model, the cluster centers can simply select a communication link for each neighboring cluster to communicate individually with the neighboring cluster centers \cite{Awerbuch85}. However, in the CONGEST model, communicating information about neighboring clusters to the cluster center might lead to congestion problems. Using a slightly more careful analysis, the procedure can be adapted to the CONGEST model.

\begin{lemma}[Implicit in \cite{Awerbuch85}]\label{lm:synchronizer}
Given a $T(n)$-round synchronous CONGEST model algorithm for constructing a sparsified $(\zeta,\delta)$-low diameter decomposition, any synchronous $R(n)$-round $M(n)$-message complexity CONGEST model algorithm can be run on an asynchronous CONGEST network with a total of $O(M(n)+R(n)(\zeta+n) )$ messages, and, provided that each message incurs a delay of at most one time unit, in time $O(R(n)\delta)$. The initialization phase takes $O(T(n))$ time, using $O((T(n)+\delta)m)$ messages.  
\end{lemma}

For a sketch of the algorithm, we refer to Appendix~\ref{app:synchronizer}. If we plug in our clustering, we obtain the following theorem. 

\begingroup
\def\thetheorem{\ref{thm:synchronizer}}
\begin{theorem}[Restated]
There exists an algorithm that, given $k\geq 2$, can run any synchronous $R(n)$-round $M(n)$-message complexity CONGEST model algorithm on an asynchronous CONGEST network. In expectation, the algorithm uses a total of $O(M(n)+R(n)n^{1+1/k})$ messages. Provided that each message incurs a delay of at most one time unit, it takes $O(R(n)k)$ rounds. The initialization phase takes $O(k)$ time, using $O(km)$ messages. 
\end{theorem}
\addtocounter{theorem}{-1}
\endgroup

\section{Constructing a Low Diameter Decomposition}
\label{sec:LDD}
In this section, we show that for an integer weighted graph the computed clustering is a probabilistic low diameter decomposition. To be precise, if we set $p=\frac{\beta}{4}$, and $r=\left\lceil\frac{1}{p} \ln\left(\frac{n^2}{p}\right)+\frac{1}{4p}\right\rceil$ we obtain a $\left(\beta,O\left(\frac{\log n}{\beta}\right)\right)$-low diameter decomposition. By Corollary~\ref{cor:strongdiameter}, we know that each of the clusters has a strong diameter of at most $2r =2\left\lceil\frac{1}{p} \ln\left(\frac{n^2}{p}\right)+\frac{1}{4p}\right\rceil = O\left(\frac{\log n}{\beta}\right)$. Now, we show that the probability that an edge $e\in E$ is an inter-cluster edge is at most $4p\cdot w(e)=\beta\cdot w(e)$. We use a general proof structure from \cite{MPX13}, but make it more streamlined; we avoid an artificially constructed `midpoint' on the edge $(u,v)$. Further, our proof borrows the idea of conditioning on the event $E_{u',v',\alpha}$ from Xu \cite{Xu17}, which we adapt to our situation. 

\begin{lemma}
For $(u,v)\in E$, the probability that $u$ and $v$ belong to different clusters is at most $4p\cdot w(u,v)$. 
\end{lemma}
\begin{proof}
Suppose $(u,v)$ is an inter-cluster edge. Without loss of generality, we assume $d_{G'}(s,v) \leq d_{G'}(s,u)$. By the triangle inequality, we have $d_{G'}(s,u) \leq d_{G'}(s,v) + w(u,v)$, hence we have $d_{G'}(s,u)-d_{G'}(s,v) \leq w(u,v)$. Using this, we can upper bound the probability that an edge $(u,v)$ is an inter-cluster edge by the probability that this inequality holds. Note that we can assume $4p\cdot w(u,v)<1$, otherwise the statement is trivially true. 

We want to condition on the cluster center $v'$ that satisfied $d^{(v')}(s,u)=d_{G'}(s,v)$, and on the cluster center $u'\neq v'$ that minimizes $d^{(u')}(s,u)$.  Moreover, we ask that these cluster respect the tie-breaking rule, i.e., both have minimal $\ID$ among all centers with equal distance. Finally, we condition on the value of $d_{G'}(s,u)$, which we set to $\alpha$. We call this event $E_{u',v',\alpha}$, which we formally define to hold when the following four conditions are satisfied
\begin{enumerate}
\item $d^{(v')}(s,u) \leq\alpha$;
\item for $w'\in V\setminus\{v'\}$ we either have $d^{(w')}(s,v)>d^{(v')}(s,v)$, or we have $d^{(w')}(s,v)=d^{(v')}(s,v)$ and $\ID(v')<\ID(w')$;\label{cond:notv'}
\item $d^{(u')}(s,v) = \alpha$;\label{cond:u'=alpha}
\item for $w'\in V\setminus\{u',v'\}$ we either have $d^{(w')}(s,u)>d^{(u')}(s,u)$, or we have $d^{(w')}(s,u)=d^{(u')}(s,u)$ and $\ID(u')<\ID(w')$.\label{cond:notu'v'}
\end{enumerate}
Now, we condition on $E_{u',v',\alpha}$ and use the law of total probability:
\begin{align*}
&\P\left[ (u,v) \text{ is an inter-cluster edge}\right] \\
&= \sum_{u' \in V} \sum_{v' \in V\setminus\{u'\}} \sum_\alpha \P\left[  (u,v) \text{ is an inter-cluster edge} \;\middle|\; E_{u',v',\alpha}\right]\P\left[ E_{u',v',\alpha}\right]\\
&\leq \sum_{u' \in V} \sum_{v' \in V\setminus\{u'\}} \sum_\alpha 2\P\left[ d^{(u')}(s,u)-d^{(v')}(s,v) \leq w(u,v) \;\middle|\; E_{u',v',\alpha}\right]\P\left[ E_{u',v',\alpha}\right].
\end{align*}
For simplicity, we omit the bounds for the sum over $\alpha$, which is a finite sum as we always have $\alpha \leq r+ m$. The factor two appears because the event assumes $d_{G'}(s,v) \leq d_{G'}(s,u)$, hence we gain a factor two by symmetry of $u$ and $v$. We look at the first probability more closely. We can loosen some of the event's restrictions, and just maintain $d^{(v')}(s,v) \leq \alpha$, as the event we examine is independent of conditions~\ref{cond:notv'}, \ref{cond:u'=alpha}, and~\ref{cond:notu'v'} of the event $E_{u',v',\alpha}$. We obtain
\begin{align*}
&\P\left[d^{(u')}(s,u)-d^{(v')}(s,v) \leq w(u,v) \;\middle|\; E_{u',v',\alpha}\right]\\
& = \P\left[  \alpha - d^{(v')}(s,v) \leq w(u,v) \;\middle|\;  d^{(v')}(s,v') \leq \alpha\right] \\
&= \P\left[ \delta_{v'} \leq r + d_G(v',v)-\alpha +w(u,v) \;\middle|\; \delta_{v'} \geq r + d_G(v',v)-\alpha\right],
\end{align*}
where the last equality holds by definition of $d^{(v')}(s,v)$.
Now if $r+d_G(v',v)-\alpha +w(u,v) \leq r$ (or equivalently, $\alpha \geq d_G(v',v)+w(u,v)$), we stay away from our cap on the geometric distribution, and hence we can apply the memoryless property of the geometric distribution to obtain
\begin{align*}
&\P\left[ \delta_{v'} \leq r + d_G(v',v)-\alpha +w(u,v) \;\middle|\; \delta_{v'} \geq r + d_G(v',v)-\alpha\right]\\
 &= 1-(1-p)^{w(u,v)}\\
& \leq p w(u,v),
\end{align*}
where the last step holds by Bernoulli's inequality. If we have $r+d_G(v',v)-\alpha +w(u,v) > r$ (or equivalently, $\alpha < d_G(v',v)+w(u,v)$), we show that the probability $\P\left[ E_{u',v',\alpha}\right]$ of the event taking place is already small:
\begin{align*}
\P\left[ E_{u',v',\alpha}\right] &\leq \P\left[d^{(v')}(s,v) \leq \alpha \text{ and } d^{(u')}(s,u) = \alpha\right] \\
&= \P\left[ \delta_{v'} \geq r +d_G(v',v) -\alpha\right]\P\left[d^{(u')}(s,u) = \alpha\right] &\text{(since the events are independent)}\\
 &\leq \P\left[\delta_{v'} \geq r-w(u,v)\right]\P\left[d^{(u')}(s,u) = \alpha\right],
\end{align*}
where the last equality holds as $r+d_G(v',v)-\alpha +w(u,v) > r$. We bound this probability as follows:
\begin{align*} 
\P\left[\delta_{v'} \geq r- w(u,v)\right]\P\left[d^{(u')}(s,u) = \alpha\right] \leq (1-p)^{r- w(u,v)}\P\left[d^{(u')}(s,u) = \alpha\right].
\end{align*}
Now we use that $r= \left\lceil \frac{1}{p} \ln\left(\frac{n^2}{p}\right)+\frac{1}{4p} \right\rceil$, to obtain
\begin{align*}
(1-p)^{r- w(u,v)} &\leq (1-p)^{ \frac{1}{p} \ln\left(\frac{n^2}{p}\right)+\frac{1}{4p}- w(u,v)}\\
&\leq \left((1-p)^{1/p}\right)^{\ln\left(\frac{n^2}{p}\right)} & (\text{as }4p\cdot w(u,v)\leq 1)\\
&\leq \frac{p}{n^2}.
\end{align*}
Combining all of this, we obtain 
\begin{align*}
&\P\left[ (u,v) \text{ is an inter-cluster edge}\right] \\
\leq &2\sum_{u' \in V}  \sum_{v' \in V\setminus\{u'\}} \sum_{\alpha \geq d_G(v',v)+w(u,v)} \P\left[ d^{(u')}(s,u)-d^{(v')}(s,v) \leq w(u,v) \;\middle|\; E_{u',v',\alpha}\right]\P\left[ E_{u',v',\alpha}\right]\\
&+ 2\sum_{u' \in V}  \sum_{v' \in V\setminus\{u'\}}\sum_{\alpha < d_G(v',v)+w(u,v)}  \P\left[ d^{(u')}(s,u)-d^{(v')}(s,v) \leq w(u,v) \;\middle|\; E_{u',v',\alpha}\right]\P\left[ E_{u',v',\alpha}\right]\\
\leq &2\sum_{u' \in V} \sum_{v' \in V\setminus\{u'\}}  \sum_{\alpha \geq d_G(v',v)+w(u,v)} p\cdot w(u,v) \P\left[ E_{u',v',\alpha}\right]\\
&+ 2\sum_{u' \in V}   \sum_{v' \in V\setminus\{u'\}}  \sum_{\alpha < d_G(v',v)+w(u,v)} \frac{p}{n^2}\P\left[d^{(u')}(s,u) = \alpha\right]\\
\leq &2 p\cdot w(u,v)\sum_{u' \in V} \sum_{v' \in V\setminus\{u'\}}  \sum_{\alpha}\P\left[ E_{u',v',\alpha}\right]+2\frac{p}{n^2}\sum_{u' \in V} \sum_{v' \in V\setminus\{u'\}}\sum_{\alpha}\P\left[d^{(u')}(s,u) = \alpha\right].
\end{align*}
Next, we notice that all events $E_{u',v',\alpha}$ are disjoint by design, so 
\begin{align*}
\sum_{u' \in V} \sum_{v' \in V\setminus\{u'\}}  \sum_{\alpha}\P\left[ E_{u',v',\alpha}\right]=1.
\end{align*}
Clearly we have $\sum_{\alpha}\P\left[d^{(u')}(s,u) = \alpha\right] =1$, as this is just a sum over all possible values of $d^{(u')}(s,u)$. Filling both in, we conclude
\begin{align*}
\P\left[ (u,v) \text{ is an inter-cluster edge}\right] &\leq 2 p\cdot w(u,v)+2\sum_{u' \in V} \sum_{v' \in V\setminus\{u'\}}\frac{p}{n^2} \leq  4p\cdot w(u,v).
\qedhere
\end{align*}
\end{proof}

Together with Corollary~\ref{cor:strongdiameter}, this gives us the following theorem. 

\begingroup
\def\thetheorem{\ref{thm:LDD}}
\begin{theorem}[Restated]
There exists an algorithm, such that for each graph $G=(V,E)$ with integer weights $w\colon E \to \{1,\dots, W\}$ and parameter $\beta\in(0,1]$ it outputs a low diameter decomposition, whose components are clusters of strong diameter of at most $O\left(\frac{\log n}{\beta}\right)$. Moreover, each edge is an inter-cluster edge with probability at most $\beta\cdot w(u,v)$. 
The algorithm runs in $O\left(\frac{\log n}{\beta}\right)$ rounds in the CONGEST model, and in $O\left(\frac{\log n \log^*n}{\beta}\right)$ depth and $O(m)$ work in the PRAM model.
\end{theorem}
\addtocounter{theorem}{-1}
\endgroup

\section{Conclusion}
We have presented an algorithm that computes a clustering, more precisely, a tree-supported sparsified low diameter decomposition. This directly leads to a sparse spanner and can be applied to compute a synchronizer for the CONGEST model. Moreover, we show that we also improve upon the state-of-the art for low diameter decompositions. By showing that clustering can be done using a capped geometric distribution, we improve on existing algorithms for spanners and low diameter decompositions in two ways. First, we obtain bounds on the diameter/stretch and running time that are independent of the random choices of the algorithm. Second, the discreteness of the geometric distribution fits the discrete nature of graph theoretical problems better than a continuous distribution. We believe this leads to a more intuitive algorithm. 

A natural question that remains is whether it would be possible to give a with-high-probability bound on the total number of inter-cluster edges or the size of the spanner rather than an in-expectation bound. A more ambitious goal is to develop a completely deterministic algorithm with the same bounds, improving on the work of Ghaffari and Kuhn \cite{GK18}. 
\bibliography{main-arxiv}
\appendix

\section{Using Sparsified Low Diameter Decompositions for Synchronization}\label{app:synchronizer}
In the following, we turn to the algorithm realizing Lemma~\ref{lm:synchronizer}, i.e., we show how we can run a synchronous CONGEST algorithm on an asynchronous CONGEST network, using a sparsified low diameter decomposition. Hereto, we present an implementation of the synchronizer $\gamma$ in the CONGEST model, using sparsified low diameter decompositions for the communication. We refer to \cite{Awerbuch85} for a proof of correctness of the synchronizer $\gamma$. 

The initialization phase consists of three steps. First, we compute the sparsified $(\zeta,\delta)$-low diameter decomposition. To do this in the asynchronous CONGEST model, we use the synchronizer $\alpha$ (for details see \cite{Awerbuch85}, or textbooks, e.g., \cite{Lynch96,KS11}). Hence this takes $O(T(n))$ time, and $O(T(n)m)$ messages. 
Second, we pick a cluster center for each cluster and construct a tree rooted at the cluster center spanning the cluster. We can do this in $O(\delta)$ time, using $O(\delta m)$ messages, again using the synchronizer $\alpha$. Note that if the computed decomposition was tree-supported, these trees are already given. As a third and final step of the initialization phase, each vertex needs to be aware of its incident sparsified inter-cluster edges, as it will use these to communicate to neighboring clusters. This might be already determined during the construction of the clustering. It could also be the case that for each sparsified inter-cluster edge, only one of the two incident vertices knows this. In $O(1)$ time, using $O(m)$ messages, this can be communicated using the synchronizer $\alpha$. 
In total, we use $O(T(n))$ time for the initialization phase, and $O((T(n)+\delta)m)$ messages. 

Now we are set up for the simulation of the $R(n)$-round, $M(n)$-message complexity synchronous CONGEST model algorithm. In each simulation of a synchronous round of this algorithm, vertices respond to each message with an acknowledge message, same as in the synchronizers $\alpha$ and $\beta$. When a vertex has received acknowledge messages for each sent message, it declares itself safe. If a vertex and all its children in the cluster tree are safe, it notifies its parent in the cluster tree. Once the cluster center has received confirmation that the whole cluster is safe, it down-casts this information to the whole cluster. Each vertex communicates that its cluster is safe over its sparsified inter-cluster edges. Once a vertex received a message of being safe over each sparsified inter-cluster edge, it declares itself \emph{ready}. When a vertex and all its children in the cluster tree are ready, it sends a ready-message to its parent in the cluster tree. Once a cluster center received ready-messages from the whole cluster, it down-casts a message `cluster ready' to all cluster vertices. 

Assuming that each message incurs a delay of at most one time unit, we need at most $O(\delta)$~time for this procedure, as we send information along the trees of height $\delta$ for a total of four times. See \cite{Awerbuch85} for the argument explaining why confirmation that neighboring clusters are done suffices. Moreover, the only communication links participating in this procedure, are the edges from the sparsified low diameter decomposition (consisting of at most $\zeta$ inter-cluster edges and $n$ tree edges). Each of these edges sends up to four messages in total, giving a total bound on the number of messages of $O(R(n)(\zeta+n) )$ for the synchronization. This gives a total time bound of $O(R(n)\delta)$, and message complexity bound of $O(M(n)+R(n)(\zeta+n) )$.

\end{document}